\documentclass[a4paper,UKenglish,cleveref, autoref,numberwithinsect]{lipics-v2019}

\title{Parameterized Fine-Grained Reductions}

\titlerunning{}

\author{Elli Anastasiadi}{National Technical University of Athens, Greece}{ellie@corelab.ntua.gr}{}{}

\author{Antonis Antonopoulos}{National Technical University of Athens, Greece}{aanton@corelab.ntua.gr}{}{}

\author{Aris Pagourtzis}{National Technical University of Athens, Greece}{pagour@cs.ntua.gr}{}{}

\author{Stavros Petsalakis}{National Technical University of Athens, Greece}{stpetsalakis@corelab.ntua.gr}{}{}

\authorrunning{E. Anastasiadi, A. Antonopoulos, A. Pagourtzis, S. Petsalakis}

\Copyright{Ellie Anastasiadi, Antonis Antonopoulos, Aris Pagourtzis, Stavros Petsalakis}



\category{}

\relatedversion{}

\nolinenumbers 

\hideLIPIcs  

\usepackage{tikz}
\usetikzlibrary{arrows.meta}
\usetikzlibrary{shapes}

\begin{document}

\maketitle

\begin{abstract} 
During recent years the field of fine-grained complexity has bloomed to produce a plethora of results, with both applied and theoretical impact on the computer science community. The cornerstone of the framework is the notion of fine-grained reductions, which correlate the exact complexities of problems such that improvements in their running times or hardness results are carried over. We provide a parameterized viewpoint of these reductions (PFGR) in order to further analyze the structure of improvable problems and set the foundations of a unified methodology for extending algorithmic results. In this context, we define a class of problems (FPI) that admit fixed-parameter improvements on their running time. As an application of this framework we present a truly sub-quadratic fixed-parameter algorithm for the orthogonal vectors problem. Finally, we provide a circuit characterization for FPI to further solidify the notion of improvement. 
\end{abstract}

\section{Introduction}
Fine-Grained Complexity deals with the exact complexity of problems, and establishes a web of refined reductions, that preserve exact solving times. While many of the key ideas come from well-known frameworks (NP-completeness program, parameterized algorithms and complexity, etc.), this significant new perspective has emerged only recently \cite{Wil15}.

The main question posed by this new field is ``\textit{given a problem known to be solvable in $t(n)$ time, is there an $\varepsilon>0$ such that it can be solved in $t^{1-\varepsilon}(n)$?}''. In the case such a result exists, we can connect this improvement to algorithmic advances among different problems. In the case it does not, we can establish conditional lower bounds based on this hardness, as is usually done with popular conjectures \cite{AWY18,WW18}. Several such conjectures are used, such as the Orthogonal Vectors Conjecture (OVC), the Strong Exponential Time Hypothesis (SETH), the APSP conjecture etc. It has been shown that OVC is implied by SETH \cite{Wil05}, and the variants and consequences of both conjectures have been extensively studied.

A main difference between the fine-grained approach and classical complexity theory is that all NP-complete problems form an equivalence class modulo polynomial-time reductions. On the contrary, fine-grained reductions produce a much more complex web. Many problems stem from SETH/OVC, others from the 3-SUM conjecture (especially Computational Geometry problems \cite{GAJ012}), and very few equivalence classes are known (a significant exception is the equivalence class for APSP \cite{AGW15}, \cite{WW18}). These observations raise questions concerning the structural complexity of fine-grained reducibility, as has traditionally been the case in other fields of complexity theory: Conditional irreducibility results, the morphology of the equivalence classes formed, fine-grained completeness notions, consequences of partitioning problems into classes etc., propose a fine-grained structural complexity program.

On the other hand, the parameterized point of view is dominant in theoretical computer science during the last decades. ETH and SETH were introduced in that context, and used widely to establish conditional lower bounds (SETH-hardness). Additionally, Fixed Parameter Tractability (FPT)\footnote{A problem is called Fixed Parameter Tractable, if there is a parameterization $k$ such that the problem can be solved in time $f(k)\cdot poly(n)$, for a computable function $f$.}, gave a multi-variable view of complexity theory, as well as the means to \emph{concentrate the hardness} of problems to a certain parameter, instead of the input size: many problems have significant improvements on their complexity, if one restricts them to instances having fixed $k$, where $k$ is a parameter of the aforementioned problems. This can be viewed as an indication of the structural importance of $k$ in each problem.

Similar techniques can be used to differentiate versions of problems, as has been seen recently in the case of fine-grained conjectures (e.g. problems on sparse graphs \cite{GIKW17}).

\subsection{Motivation}

The conditional bounds shown by fine-grained reductions stem from relating improvements between the conjectured best running times of two problems. This has resulted in an effort to classify problems either through equivalence or hardness via a minimal element (OV-hardness) \cite{CW19}.

Additionally, most known fine-grained reductions inherently relate the problems in more than trivial ways, also mapping specific parameters of each problem to one another. This could indicate relations between the problems' property of concentrating hardness to certain parameters.

On the other hand, while parameterized complexity has traditionally been concerned with breaching the gap between polynomial and exponential running times, there has recently been interest in fixed-parameter improvements amongst polynomial time problems (sometimes referred to as FPT in P \cite{Wil15,GMN15}). 

The ability of fine-grained complexity to express correlations between problems of various running times, comes with some inherent theoretical obstacles. Specifically, the new viewpoint of a problem's "hardness" is associated with the capability to improve its running time. This results in a "counter-intuitive" notion of hard problems, as they frequently correspond to (in classical terms) easy ones. Moreover, the foundation on which this framework is based, allows the computational resources of a fine-grained reduction to change depending on the participant problems. This produces vagueness regarding what would be considered a complexity class compatible with such reductions.

Our concern is to surpass the inherent difficulties of the field towards constructive methods that produce generalizable results, as well as to contribute to the effort of establishing structural foundations for the framework. This could result in furthering our understanding of what constitutes  difficulty in computation, as well as to structurally define improvability.

\subsection{Our Results}
We introduce Parameterized Fine-Grained Reductions (PFGR), a parameterized approach to fine-grained reducibility that is consistent with known fine-grained reductions and offers (a) tools to study structural correlations between problems, and (b) an extension to the suite of results that are obtained through the reductions. This provides a multi-variate approach of analyzing fine-grained reductions. Additionally, we give evidence that these reductions connect structural properties of the problems, such as the aforementioned concentration of hardness.

We define a class of problems (FPI) that admit parameterized improvements on their respective conjectured best running time algorithms. To this end, we treat improvements in the same way as fine-grained complexity (i.e. excluding polylogarithmic improvements in the running time). This gives us the expressive power to correlate structural properties of problems that belong in different complexity classes.

We prove that this class is closed under the aforementioned parameterized fine-grained reductions, which can be used as a tool to produce non-trivial parameterized algorithms (via the reduction process). We present such an application in the case of the reduction from OV to Diameter, in which we use a fixed parameter (with respect to treewidth) algorithm for Diameter to produce a new sub-quadratic fixed parameter algorithm for OV running in time $O\left( d^2(n+d)\log^d(n+d)\right)$ where $d$ is the dimension of the input vectors.

Finally, we use notions from parameterized circuit complexity to analyze membership in this class and introduce a circuit characterization, similar to the one used in the definition of the W-Hierarchy in parameterized complexity \cite{DT11}.

\subsection{Related Work}

The fine-grained reductions literature has quickly grown over the recent years (see \cite{Wil15} for a survey).
The basis for reductions have been some conjectures that are widely considered to be true. Namely, the 3SUM, APSP, HITTING SET \& SETH conjectures, which are associated with the respective plausibility for improvement of each problem. 

A large portion of known reductions stem from the Orthogonal Vectors problem, which is known to be SETH-hard \cite{Wil05}, thus OV plays a central role in the structure of the reductions web. Different versions of the OV conjecture were studied, usually parameterized by the dimension \cite{ABDN18}. 

The logical expressibility of problems similar to OV were studied \cite{GIKW17}, research that created a notion of hardness for problems in first-order logic, and introduced various equivalence classes (\cite{GIKW17,GI19,CW19}) concerning different versions of the OV problem, with significant applications to other fields (such as Computational Geometry). Additionally, the OV conjecture was studied in restricted models, such as branching programs \cite{KW19}.

Structural implications of fine-grained irreducibility and hypotheses were studied, culminated to new conjectures, like NSETH \cite{CGIMPS16}. New implications of refuting the aforementioned hypotheses on Circuit Lower Bounds were discovered \cite{JMV15,AHWW16,AB18}, and the refutation of SETH would imply state-of-the-art lower bounds on non-uniform circuits. Recently, fine-grained hypotheses were connected to long-standing questions in complexity theory, such as Derandomization of Complexity Classes \cite{CIS18}.

The parameterized analysis of algorithms, one of the most active areas in theoretical and applied computer science, has been frequently used to provide tools in fine-grained complexity. As such, new conjectures were formed about the solvability of polynomial-time problems in terms of parameters \cite{AWW15}.

In a notable case of similar work \cite{BK18}, the authors analyze the multivariate (parameterized) complexity of the longest common subsequence problem (LCS), taking into account all of the commonly discussed parameters on the problem. As a result, they produce a general conditional lower bound accounting for the conjunction of different parameterized algorithms for LCS: Unless SETH fails, the optimal running time for LCS is $\left( n+min\{d,\delta \Delta,\delta m\}\right)^{1\pm o(1)} $ where $d,\delta,\Delta,m$ are the aforementioned parameters. 
Note that our work is in a different direction to this result. Instead of separately reducing SETH to each different parameterized case, we give the means to show correlations between parameters in such reductions, i.e. with this framework one can analyze a single  reduction to show multiple dependencies between parameterized improvements for each problem. While this automatically produces several conditional bounds among parameterizations of problems, it proves most useful in the opposite direction, namely to transfer improvements between problems and thus derive new parameterized algorithms.

\section{Preliminaries}
We denote with $[n]$, for $n\in\mathbb{N}$, the set $\{1,\ldots,n\}$.

\begin{definition}[Generalized Parameterized Languages]
Let $L \subseteq \Sigma^*$, and $k_1,\dots,k_\ell$ parameterization functions, $k_i:\Sigma^* \to \mathbb{N}$, $1\le i \le \ell$. Let $\langle L,k_1,\dots, k_{\ell}\rangle$ denote the corresponding parameterized language.
\end{definition}
For simplicity,  we will use $\langle L \rangle$ to abbreviate $\langle L,k_1,\dots, k_{\ell}\rangle$, and $I_L$ to denote an input instance for $\langle L \rangle$. 

Note here the divergence from the classical definition, that associates each problem with only one parameter \cite{DT11}. We prefer the generalized version that allows us to describe simultaneously several structural measures of the problem, such as number of variables, number of nodes and more complex ones. For each one of those parameters we assume that there exists an index $j \in \{1, \ldots \ell \}$ such that $k_j$ corresponds to the mapping of the input instance $I_L$ to this specific parameter. 
In this way we can not only isolate and analyze different characteristics of structures but also treat each of these measures individually. 

\begin{definition}[OV]
Define the Orthogonal Vectors problem ($OV$) as follows: Given two sets $A,B \subseteq \{0,1\}^d$, with $|A|=|B|=n$, are there two vectors $a\in A$, $b\in B$, such that $a\cdot b=\sum_{i=1}^d a[i]\cdot b[i] = 0$?
\end{definition}

\begin{definition}[$3/2$-Approx-Diameter]
Given a graph $G=(V,E)$, approximate its diameter, i.e. the quantity 
$\max_{u,v \in V}d(u,v)$, within a factor $3/2$.
\end{definition}

We will also define the notion of treewidth as we will later present a result that utilizes it as a graph parameter.

\begin{definition}[Treewidth]
A \textbf{tree decomposition} of a graph $G = (V, E)$ is a tree, $T$, with nodes $X_1,X_2 \ldots X_n$ (called bags), where each $X_i$ is a subset of $V$, satisfying the following properties:
\begin{itemize}
    \item The union of all sets $X_i$ equals $V$. That is, each graph vertex is contained in at least one tree node.
    \item The tree nodes containing vertex v, form a connected subtree of T.
    \item For every edge $(v, w)$ in the graph, there is a subset $X_i$ that contains both $v$ and $w$.
\end{itemize}
The \textbf{width} of a tree decomposition is the size of its largest set $X_i$ minus one. The \textbf{treewidth} $tw(G)$ of a graph $G$ is the minimum width among all possible tree decompositions of $G$. 
\end{definition}

For more information on parameterized complexity, and treewidth the reader is referred to \cite{DT11, DF13}.

We will utilize the following notions from circuit complexity theory (for more details the reader is referred to Ch. 6 of \cite{AB09}).

\begin{definition}[Circuit Complexity] 

The circuit-size complexity of a Boolean function ${\displaystyle f} $ is the minimal size (number of gates) of any circuit computing ${\displaystyle f}$. The circuit-depth complexity of a Boolean function ${\displaystyle f}$ is the minimal depth of any circuit computing ${\displaystyle f}$.
\end{definition}

We will also use the following definition of Fine-Grained reductions from \cite{WW18}. 
\begin{definition}[Fine-Grained Reduction]
\label{FGR}
 Let $a(n), b(n)$ be nondecreasing functions of $n$.
 Problem $A$ is $(a,b)$-reducible to problem $B$ (denoted $A \le_{FG} B$), if for all $\varepsilon>0$ there exists a $\delta >0$, and an algorithm $F$ solving $A$ with oracle access to $B$ such that $F$ runs in at most $d\cdot a^{1-\delta}(n)$ time, making at most $k(n)$ oracle queries adaptively (i.e. the $j^{th}$ instance $B_{j}$ is a function of $\{B_i,a_i\}_{1\le i<j}$). The sizes $|B_{i}|=n_i$ for any choice of oracle answers $a_i$, obey the inequality:	$$ \sum_{i=1}^{k(n)}b^{1-\varepsilon}(n_i) \le d \cdot a^{1-\delta}(n)$$ 
\end{definition}

\section{Parameterized Fine-Grained Reductions}

In this section we define Parameterized Fine-Grained Reductions (PFGR), along with some examples of applications, and show their relation to fine-grained reductions.

\begin{definition}[PFGR]
Given problems $A$ and $B$ with $a(n)$, $b(n)$ their respective conjectured best running times: We say $\langle A,k_1,\dots, k_{i_A} \rangle \le_{PFG} \langle B,\lambda_1,\dots, \lambda_{i_B} \rangle$ if there exists and algorithm R such that
\begin{enumerate} 
\item For every $\varepsilon>0$ there exists a $\delta>0$ such that R runs in $a^{1-\delta}(n)$ time on inputs $I_A$ of length $n$ by making $q$ query calls to $\langle B \rangle$ with query lengths $n_1,\dots,n_q$, and $\sum_{i=1}^q b^{1-\varepsilon}(n_i) \le c\cdot a^{1-\delta}(n)$, for some constant $c>0$, and  R accepts iff $I_A \in \langle A \rangle$.

\item For every query $q_j ,j=1, \ldots, q$, there exists a computable function $g_j:\mathbb{N}^{i_A} \to \mathbb{N}^{i_B
}$ defined as $g_j(k_1, \ldots, k_{i_A})= [g_{j,1},g_{j,2},\ldots, g_{j,i_B}]$  such that  for every $ \lambda_i \in \langle \lambda_1,\dots, \lambda_{i_B} \rangle$\\$ \lambda_i \le g_{j,i}(k_1, \ldots, k_{i_A})$.
\end{enumerate}
\end{definition}

\begin{remark}
The number of calls is specific to the type of reduction used. In the case of adaptive queries, the number of potential calls could exceed $q$ exponentially (in the worst case).
A reasonable objection could arise here since the mappings $g_j$ are not assumed to have any time restriction.
However, $g_j$ are not implemented by the reduction algorithm, and are merely correlating functions of the parameters. As such, they do not affect the running time of the reduction.
\end{remark}

Additionally, one would suspect that this definition is a limitation on the original fine-grained framework and hence is only satisfied by some of the known reductions. The main problem is that most of the known reductions refer to non parameterized problems. This however can be easily surpassed by our formalization, as we can view these as projections of PFGR (see section \ref{equiv}): 

Given a problem $P$ and an input instance $I_P$, the parameterized version of the problem can be produced by extending the input with the computable function that defines each parameter over it. We can now redefine any reductions it took place in, simply replacing the problem with its parameterized version. This essentially provides us with all of the possible parameterizations a problem can have, and uses them as a whole in order to preserve structural characteristics. 

While analyzing reductions, some notable cases occur:
Firstly, the case of only one query call, as observed in the majority of known fine-grained reductions. Secondly, the case where even though many query calls are made, the constructions of the input instance to $\langle B\rangle$ maintain uniform mappings of the parameters, i.e. $g_j=g_{j'}, \forall j,j' \in [q]$. Lastly, the case where the value of each parameter of problem $\langle B\rangle$ is only related to a single parameter of $\langle A\rangle$, i.e. $g_{j,i}:\mathbb{N} \to \mathbb{N}$. \footnote{this case is especially useful in transferring parameterized improvements, as we will see in section \ref{closure}.} 

We provide some examples to further clarify our definition and notation:
\begin{example}
Consider the well-studied reduction $CNF\mbox{-}SAT \le_{2^n,n^2}k\mbox{-}OV$, presented in \cite{Wil05}. It is apparent that through one call, the number of clauses $m$ of the SAT instance corresponds to the dimension $d$ of the OV  vectors instance, as well as that the number of variables $n$ is mapped to the number of vectors $N$ via the mapping:
$g(n,m)=\langle 2^{n/k},m \rangle=\langle N,d\rangle=I_{OV}$.
This means that the input instance to OV will contain $N=2^{n/k}$ vectors of dimension $d=m$.
\end{example}
This procedure is summarized in the first row of the following table, as well as other indicative reductions, in the same context. For a more detailed analysis of each reduction see the full version.

\begin{table}[h!]
\[
\begin{array}{|l|l|c|c|}
\hline
\mbox{\textbf{Reduction}}   & \mbox{\textbf{Mapping}}                             & \mbox{\textbf{calls}} & \mbox{\textbf{Ref.}}        \\ \hline
I_{SAT} \le I_{OV} & g(n,m): N = 2^{n/k}, d = m & 1                & \cite{Wil05} \\ \hline
I_{APSP} \le I_{\text{MPProd}} & g(\text{Nodes}) : (n \times n , n \times n) & \lceil \log n \rceil & \cite{WW18} \\ \hline
I_{\text{MPProd}}\le I_{APSP} & g(n_1 \times n_2,n_2 \times  n_3): (\text{Nodes} = n_1+n_2+n_3) & 1 & \cite{WW18}\\ \hline
I_{\text{MPProd}}\le I_{\text{NegTr}} & g( n_1 \times n_2, n_2 \times n_3): (\text{Nodes} = n_1+n_2+n_3) & \log n &\cite{WW18}\\ \hline
I_{\text{NegTr}}\le I_{\text{Radius}} & g(\text{Nodes} = n,\text{Weights} = n^c): (4n,3n^c) & 1 & \cite{AGW15} \\  \hline
I_{APNT} \le I_{\text{NegTr}}  & g(\text{Nodes} = n) : (\text{Nodes'} = \sqrt[3]{n}) & n^2+\frac{n^3}{\sqrt[3]{n}} & \cite{WW18}\\ 
\hline 
\end{array}
\]
\caption{Some reductions using our notation (\textit{where APNT abbreviates the All-Pairs Negative Triangle problem, NegTr the Negative Triangle, and MPProd the Min-Plus product problem.})}\label{table1}
\end{table}

\paragraph*{Consistency with fine grained complexity}
\label{equiv}
While our framework encapsulates many natural structural properties, there are problems that are fine-grained reducible to each other and either do not have an obvious correlation between their structures, or have connections that are not apparent. We will show here that our definition for Parameterized Fine-Grained Reduction is consistent with those cases, as the set of problems that are reducible to each other via fine grained reductions (denoted as $S_1$) and the respective set for Parameterized Fine Grained Reductions (denoted as $S_2$) are equivalent.

\begin{theorem}
Let $S_1:=\{(A,B): A\le_{FG}B \}$ and $S_2:=\{(A,B): A\le_{PFG}B\}.$\\ 
Then $S_1=S_2$.
\end{theorem}

\begin{itemize}

\subparagraph*{}
    \item \textbf{{$S_2 \subseteq S_1$}}\\
Firstly, given problems $A$ and $B$ with $a(n)$, $b(n)$ their respective conjectured best running times, if $\langle A \rangle \le_{PFG} \langle B \rangle $, we can simply ignore the parameters involved in the reduction and treat it as a fine-grained reduction between A and B, as the time restrictions enforced in both definitions are identical ($a^{1-\delta}(n)$ bound for the reduction time, and $\sum_{i=1}^{q} b^{1-\varepsilon}(n_i)$ for the calls to problem $B$). 

\subparagraph*{}
\item \textbf{{$S_1 \subseteq S_2$}}
\begin{lemma}
Given problems $A$ and $B$ with $a(n)$, $b(n)$ their respective conjectured best running times, if $A\le_{FGR}B$, for every $\lambda_i$ in a given parameterization $\langle B \rangle $, for each query call $q_j$ made in the reduction, there exists a computable function $g_{j,i}:\mathbb{N}^{i_A}\longrightarrow\mathbb{N}$ such that $\lambda_i \le g_{j,i}(k_1,\ldots,k_{i_A})$, and as such $\langle A \rangle \le_{PFG} \langle B \rangle$.
\end{lemma}

\begin{proof}
We remind the reader here that for $\lambda_i$ to be considered a parameter of a problem $A$, it has to be the output of a computable function on the input of $A$. Hence, every parameter of $B$ has a computable function $f_i(I_B)$ associated with it. Now, since the reduction producing the instances of $B$ is $a^{1-\delta}$-time computable, it can trivially be viewed as a computable function $F$ having as domain field the inputs to problem $A$, and range the input instances of $B$ it produces. Having these, we can simply take the composition of $F$ and each $f_i$ to produce computable functions $f'_i=F\circ f_i$ that produce the aforementioned parameters of $B$. Hence, these parameters can be viewed both as parameters of $B$ and parameters of $A$. For these reasons, the fine-grained reduction can be viewed as $\langle A,k_1,\dots\rangle \le_{PFG} \langle B,\lambda_1,\dots \rangle$ for $k_i=\lambda_i$ (ergo having the identity function as $g_{j,i}$).
\end{proof}

\end{itemize}

\section{Fixed Parameter Improvable Problems (FPI)}

In this section we define a class of problems that admit parameterized  improvements on their conjectured best running times,  prove that this class is closed under PFGR, as well as produce new parameterized improvements as an application of this closure.

 \begin{definition}[FPI]
 Let $A$ be a problem with conjectured best running time $a(n)$.
 Then, $\langle A \rangle $ has the FPI property with respect to a set of parameters $K=(k_1,k_2, \ldots, k_x) \subseteq \langle A \rangle$ \footnote{in this context, $K\subseteq \langle A \rangle$ denotes a set of parameterization functions over the input of $A$.} (denoted
 FPI$(A,K)$) if there exists an algorithm solving $\langle A \rangle $ in $O\left( a^{1-\varepsilon}(n)\cdot f(k_1,k_2, \ldots, k_x)\right)$ time, for some $\varepsilon>0$ and a computable function $f$. 
 \end{definition}
 
 \begin{remark}
 For simplicity, in the case of a single parameter, we denote as $FPI(A,k)$ the property $FPI(A,\{k\})$.
 \end{remark}
 
 \begin{theorem}\label{NPFPI}
For every NP-hard problem $A$ that admits an FPT algorithm w.r.t.\  a parameter $k$, we have that $FPI(A,k)$, unless $P=NP$.
 \end{theorem}

\begin{proof}
Since all NP-hard problems are conjectured to demand exponential running time, any FPT algorithm that solves them in $O\left( n^c\cdot f(k)\right)$ time for some parameter $k$ can be viewed as an improvement to $O(a^{1-\varepsilon}(n)\cdot f(k))$; the actual improvement in the conjectured running time $a(n)$ is in fact much greater than $a^{\varepsilon}(n)$ .  
\end{proof}

\begin{corollary}
The following problems, parameterized with the respective parameter are FPI:
\begin{itemize}
    \item $\langle$ Vertex Cover, Solution size $\rangle$
    \item $\langle$ SAT, Number of clauses $\rangle$
    \item $\langle$ k-knapsack,k $\rangle$
\end{itemize}
\end{corollary}

\begin{definition}[Minimum Necessary Set]
Let $\langle A,k_1,\dots, k_{i_A} \rangle$,  $\langle B,\lambda_1,\dots, \lambda_{i_B} \rangle$ be parameterized problems such that $\langle A,k_1,\dots, k_{i_A} \rangle \le_{PFG} \langle B,\lambda_1,\dots, \lambda_{i_B} \rangle$.
We define as $K_\Lambda$ to be the minimum necessary set needed to bound the set of parameters $\Lambda$ of $\langle B \rangle$ with respect to g, ergo the parameter set ${k_1,\ldots,k_x}$ of problem A for which $\exists \lambda_i \in  \Lambda, \forall j \in[q]$ such that $g_{j,i}(k_1,\dots, k_{A}) \le h(k_1,\ldots,k_x),~ k_i \in K_\Lambda$ for some computable function $h$.
\end{definition}

\begin{theorem}[Closure under PFGR]\label{closure}
Let $\langle A,k_1,\dots, k_{i_A} \rangle$,  $\langle B,\lambda_1,\dots, \lambda_{i_B} \rangle$ be parameterized problems. If $\langle A,k_1,\dots, k_{i_A} \rangle \le_{PFG} \langle B,\lambda_1,\dots, \lambda_{i_B}\rangle$ and $FPI(B,\Lambda)$, then $FPI(A,K_\Lambda)$, where $K_\Lambda$ is the minimum necessary set of $\Lambda$ w.r.t. $g$.
\end{theorem}

\begin{proof}
It suffices to prove that there exists an algorithm for $A$ running in $O(a^{1-\varepsilon}(n)\cdot f(k_1,\ldots,k_x))$ for $k_i,~i\in [x]$, parameters of $\langle A \rangle$.

Since $\langle A \rangle \le_{PFG} \langle B \rangle$, then for all $\varepsilon>0$ there exists a $\delta>0$ such that $I_A$ is evaluated by an $a^{1-\delta}(n)$ algorithm using $B$ as an oracle, and $\sum_{i=1}^{q} b^{1-\varepsilon}(n_i) \le a^{1-\delta}(n)$. 

Also, $FPI(B,\Lambda)$, so there is an algorithm computing $B$ in $O(b^{1-\varepsilon'}(n)\cdot f(\lambda_1,\ldots,\lambda_{|\Lambda|}))$. We can use this algorithm to resolve the oracle calls in time $\sum_{i=1}^{q} b^{1-\varepsilon'}(n_i)f(\lambda_1,\ldots,\lambda_{|\Lambda|})$.

We can use $K_\Lambda$ to describe the running time for B utilizing function h. As such, the total running time of $A$ is:

$$a^{1-\delta}(n)+\sum_{i=1}^{q} b^{1-\varepsilon'}(n)f(\lambda_1,\ldots,\lambda_{|\Lambda|})
= a^{1-\delta}(n)+\sum_{i=1}^{q} b^{1-\varepsilon'}(n)f(h(k_1,\ldots,k_{|K_\Lambda|})) 
$$

$$\le a^{1-\delta}(n) \cdot f'(k_1,\ldots,k_{|K_\Lambda|})
$$
Hence, $FPI(A,K_\Lambda)$.
\end{proof}

The above result essentially means that any parametric improvement can be carried through a valid PFGR. 

\subsection{A subquadratic fixed-parameter algorithm for OV}

We will now provide an analysis of a known reduction from $OV$ to $3/2$-approx-diameter \cite{RW13} using the PFGR framework. Specifically, we show that since $3/2$-approx-diameter admits fixed parameter improvements \cite{AWW15} on the treewidth parameter (hence is in FPI), this can be used to provide a fixed-parameter improvement on the OV problem.

\begin{theorem}\label{ovtodiam}
$\langle OV,dimension\rangle$ is PFG-reducible to $\langle 3/2$-approx-diameter$,treewidth\rangle $.
\end{theorem}

\begin{proof}

Note that this reduction is implemented using only one call, and we analyze only one parameter. As such, we will simplify the notation of $g_{j,i}$ to $g$. 

We begin with the construction given in \cite{RW13}: 
Given an $OV$ instance with sets $A$, $B$ as input, we create a graph as follows: for every $a \in A$ create a node $a \in G$, for every $b\in B$ create a node $b\in G$, and for every $i \in [d]$ create a node $c_i\in G$, as well as two nodes $x,y$. For every $a\in A$ and $i\in [d]$, if $a[i]=1$ we add the edge $(a,c_i)$. Similarly, for every $b\in B$ and $i\in [d]$, if $b[i]=1$, we add the edge $(b,c_i)$. Also, we add the edges $(x,a)$ for every $a\in A$, $(x,c_i)$ for every $i\in[d]$, $(y,b)$ for every $b\in B$, $(y,c_i)$ for every $i\in [d]$, and $(x,y)$. 

It suffices to find which parameter is connected to treewidth via the reduction:
As seen in \ref{figureTW}-(a), the graph produced by the reduction has a very specific structure. That is, all nodes of group $A$ are linked exclusively with nodes of group $C$ and with node $x$. Similarly, for group $B$ we have connections to group $C$ and node $y$. Therefore, to produce a tree decomposition of $G$ we can leave the nodes of group $A$ unrelated with those of group $B$. 

Now, the specific connections of nodes from the groups $A$, $B$ to the nodes of group $C$ can vary, depending on the form of the OV instance. We can however give an upper bound to the treewidth of $G$ as shown in \ref{figureTW}-(b) by copying the whole group $C$ in all of the decomposition's bags. One can check that each component induced by a label is connected, and that all edges of $G$ are covered by the given bags. This decomposition is of maximum bag size $d+2$ and hence of width $d+1$, where $d$ is the size of group $C$. 

As follows from the definition, $d+1$ is an upper bound for the treewidth of $G$. However,  another tree decomposition with smaller width could exist. In order to prove that the  treewidth of $G$ is exactly $d+1$, we must show that there exist instances of OV that produce graphs through this reduction, corresponding to treewidth $d+1$. Depending on the vectors containing a 1 coordinate in the suitable position, we could end up with a graph containing as a minor a complete $K\{|A|=|B|,|C|\}$ bipartite graph. Since complete bipartite graphs $K_{m,n}$ have treewidth exactly $\min \{m,n\}$, we can deduce that all graphs $G$ produced by this reduction will have treewidth in the worst case $d+1$ (as seen in chapter 10 of \cite{DF13}).

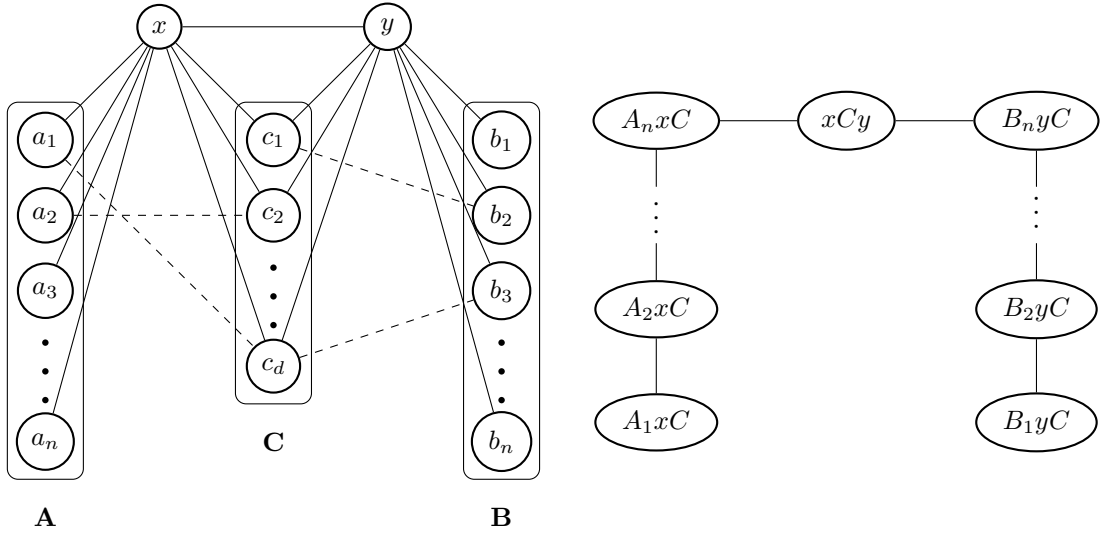
\begin{figure}[t]
\begin{subfigure}[t]{0.4\textwidth}
\centering
\begin{tikzpicture}
\begin{scope}[every node/.style={circle,thick,draw}]
    \node (L1) at (0,0) [draw=none] {\textbf{A}};
    \node (L1) at (3,1) [draw=none] {\textbf{C}};
    \node (L1) at (6,0) [draw=none] {\textbf{B}};	
    \node (A) at (0,1) {$a_n$};
    \node (B) at (0,3) {$a_3$};
    \node (C) at (0,4) {$a_2$};
    \node (D) at (0,5) {$a_1$};

    \node (E) at (3,2) {$c_d$};
    \node (F) at (3,4) {$c_2$};
    \node (H) at (3,5) {$c_1$};
    
    \node (I) at (6,1) {$b_n$};
    \node (G) at (6,3) {$b_3$};
    \node (K) at (6,4) {$b_2$};
    \node (L) at (6,5) {$b_1$};
    
    \node (X) at (1.5,6.5) {$x$};
    \node (Y) at (4.5,6.5) {$y$};
\end{scope}

\path (A) -- (B) node [font=\Huge, midway, sloped] {$\dots$};
\path (I) -- (G) node [font=\Huge, midway, sloped] {$\dots$};
\path (E) -- (F) node [font=\Huge, midway, sloped] {$\dots$};

\path (X) edge (Y);

\path (X) edge (A);
\path (X) edge (B);
\path (X) edge (C);
\path (X) edge (D);

\path (Y) edge (I);
\path (Y) edge (G);
\path (Y) edge (K);
\path (Y) edge (L);

\path (X) edge (E);
\path (X) edge (F);
\path (X) edge (H);

\path (Y) edge (E);
\path (Y) edge (F);
\path (Y) edge (H);

\begin{scope}   [dashed]
\path (C) edge (F);
\path (D) edge (E);
\path (H) edge (K);
\path (E) edge (G);

\end{scope}
\draw [rounded corners] (-0.5,0.5) rectangle (0.5,5.5);
\draw [rounded corners] (2.5,1.5) rectangle (3.5,5.5);
\draw [rounded corners] (5.5,0.5) rectangle (6.5,5.5);
\end{tikzpicture}
\caption{Graph produced by the reduction}
\end{subfigure}\hfill
\begin{subfigure}[t]{0.45\textwidth}
\centering
\begin{tikzpicture}
\begin{scope}[every node/.style={ellipse,thick,draw}]
    \node (A) at (0,1.5) {$A_1 x C$};
    \node (B) at (0,3) {$A_2 x C$};
    \node (C) at (0,5.5) {$A_{n} x C$};
    \node (D) at (5,1.5) {$B_1 y C$};
    \node (E) at (5,3) {$B_{2} y C$};
    \node (F) at (5,5.5) {$B_{n} y C$};
    \node (G) at (2.5,5.5) {$x C y$};

    \end{scope}    
\node (L1) at (0,4) [draw=none] {};
\node (L2) at (0,4.5) [draw=none] {};
\node (L3) at (5,4.5) [draw=none] {};
\node (L4) at (5,4) [draw=none] {};

\node (N1) at (0,0) [draw=none] {};

\path (B) -- (C) node [font=\Large, midway, sloped] {$\dots$};    
\path (F) -- (E) node [font=\Large, midway, sloped] {$\dots$};    
\path (A) edge (B);
\path (D) edge (E);
\path (C) edge (G);
\path (F) edge (G);
\path (B) edge (L1);
\path (C) edge (L2);
\path (F) edge (L3);
\path (E) edge (L4);
\end{tikzpicture}
\caption{Tree Decomposition}
\end{subfigure}
\caption{PFG reduction from $OV$ to $3/2$-approx-diameter.}\label{figureTW}
\end{figure}

Since $d$ is exactly the dimension of the OV instance producing the graph, we can see that there exists a function $g$ that maps the dimension of the OV instance to the treewidth of the $3/2$- approx-dimension. 

Therefore, the mapping is $g(\langle OV,d\rangle) = \langle 3/2$-approx-diameter $,treewidth-1\rangle$
\end{proof}

It was shown in \cite{AWW15} that $3/2$-approx-diameter parameterized by treewidth has a subquadratic algorithm.
Hence, FPI($3/2$-approx-diameter,treewidth). Now, as follows by theorem \ref{closure} we should expect that FPI(OV,k) where $g(k)=treewidth$. Equivalently, we would expect a parametric improvement to the running time of OV for the instances that are related to the ones of the 3/2-approx-Diameter problem of bounded treewidth.

By theorem \ref{ovtodiam} we have that since $3/2$-approx-diameter has a parameterized improvement for fixed treewidth, OV has a subquadratic algorithm for fixed dimension of the vectors. 

We will construct a subquadratic fixed-parameter algorithm, via the process described above.

Specifically, for the reduction time:
The graph constructed contained $2n+d+2$ nodes and $2nd+2d+2n$ edges, which can be constructed in $O(nd)$ time from the OV instance. 

The resulting graph has $O(n+d)$ nodes, and the decision problem of $3/2$-approx-diameter for this graph also gives an answer for the decision problem of the OV instance.

The parameterized complexity of the algorithm solving diameter is $k^2 n \log^{k-1}n$, where $k$ denotes the treewidth of the graph \cite{AWW15}.

Ergo, since $k=d+1$ (via our reduction) we can use the above to obtain an algorithm for OV running in time:
$$(d+1)^2(n+d)\log^{d+1-1}(n+d)=O\left( d^2(n+d)\log^d(n+d)\right) $$

\begin{remark}
Through our reduction, this result can be carried out to all problems PFG-reducible to OV, such as SAT or Hitting Set. See the full version for the respective analyses of these reductions.
\end{remark}

\begin{remark}
As we have seen, if problem $B$ admits parameterized improvements on parameters $\Lambda$, then through the reduction this can be translated to improvements on problem $A$ and parameters $K_\Lambda$ such that $\lambda_i \le g(k_1,...,k_{|K_\Lambda|})$, for $i\in [|\Lambda|]$. However, whether we can locate which parameters constitute $K_\Lambda$ or not depends on the invertibility of $g$. In the case $g$ is not invertible one can only show the existence of such an algorithm, but not necessarily construct it. Nevertheless, the FPI property still holds through our definition, because we can abuse the notation to interpret each $\lambda_i$ as a parameter of A, as it is a byproduct of the reduction which is \emph{an ($a(n)$-time) computable function on the input of $A$}.
\end{remark}

\section{Circuit Characterization of FPI}

We provide a characterization for FPI using circuit complexity. Specifically, it is known that any circuit of size $S(n)$ can be simulated by an algorithm with complexity $O(S(n))$, thus if one can design a circuit with size smaller than the conjectured complexity of the problem, then this can be translated into a faster algorithm.

As such, having a circuit of size $S(n)$, if we can fix any number of parameters $x$ such that the circuit can be seen as having $S'(n)\leq a(n)^{1-\varepsilon}f(k_1,\ldots,k_x)$ size, we can use this circuit to produce a truly sub-$a(n)$ algorithm for our problem.

Nevertheless, the smallest circuit solving the problem may differ from the one produced via a simulation of an algorithm\footnote{which is the only universal way to produce a circuit from an arbitrary algorithm.}. This means that an improvement in the size complexity of the circuit may not be enough to be translated into a more effective algorithm via an inverse simulation.
In that case, for the improvement in the size of the circuit to be translated to a faster algorithm, it is necessary to exceed this difference.
From now on, when referring to a circuit solving a problem, the reader should consider the one produced by the simulation procedure.

As shown in \cite{Fur82}, we can simulate any algorithm running in time $a(n)$ by a circuit of size $S(n) = a(n)log(a(n))$. 
Thus, we can use this as an upper bound on the overall size complexity of the circuit produced, to show that an improvement in the size of the circuit $S^{1-\varepsilon}(n)f(k_1,\ldots,k_x)$ (for some $x$) is always sufficient.

\begin{theorem}
Let $A$ be a problem with $a(n)$ conjectured best running time. Then, $FPI(A,K)$ if and only if for the uniform circuit family $\{C_n\}$ of size $S(n)$ computing $A$, for each $n\in\mathbb{N}$, $C_n$ has size $S'(n)=S^{1-\varepsilon}(n)\cdot f(k_1,\ldots,k_{|K|})$, for a computable function $f$.
\end{theorem}

\begin{proof} 

``$\Rightarrow$'':

\begin{eqnarray*}
    S'(n) &=& S^{1-\varepsilon}(n)f(k_1,\ldots,k_{|K|})\\
        &=& a^{1-\varepsilon}(n) ( \log a(n))^{1-\varepsilon}  \le a^{1-\varepsilon}(n) a^{\delta}(n)f(k_1,\ldots,k_{|K|})\text{, for any }\delta>0.
\end{eqnarray*}

If we choose $0<\delta<\varepsilon$, then $a^{1-\varepsilon+\delta}(n)f(k_1,\ldots,k_{|K|})=a^{1-\varepsilon'}(n)f(k_1,\ldots,k_{|K|})$ for $\varepsilon'=\varepsilon-\delta$, which is an FPI improvement on the running time of the algorithm, since we can simulate the circuit of size $S'(n)$ in linear time.

``$\Leftarrow$'':

 if there is an algorithm and a parameter set $K$ for which the running time is $a^{1-\varepsilon}(n)f(k_1,\ldots,k_{|K|})$, then we can simulate it with a circuit of size:
\begin{align*}
S'(n) &= a^{1-\varepsilon}(n)f(k_1,\ldots,k_{|K|})\log\left( a^{1-\varepsilon}(n) f(k_1,\ldots,k_{|K|})\right)\\ &= 
a^{1-\varepsilon}(n)\left(f(k_1,\ldots,k_{|K|})\log a^{1-\epsilon}(n)+f(k_1,\ldots,k_{|K|})\log(f(k_1,\ldots,k_{|K|}))\right) \\&= a^{1-\varepsilon}(n) \log a^{1-\varepsilon}(n)f'(k_1,\ldots,k_{|K|}) \leq a^{1-\varepsilon}(n) a^\delta(n) f'(k_1,\ldots,k_{|K|})\text{, for any }\delta>0.
\end{align*}
\begin{remark}
In the scope of parameterized complexity, we can transform the addition in the second line into multiplication, since it is equivalent, as seen in \cite{DF13}. 
\end{remark}

If we choose $0<\delta<\varepsilon$, then $S'(n)=a^{1-\varepsilon+\delta}(n)f'(k_1,\ldots,k_{|K|})=a^{1-\varepsilon'}(n)f'(k_1,\ldots,k_{|K|})\le S^{1-\varepsilon'}(n)f'(k_1,\ldots,k_{|K|})$, for $\varepsilon'=\varepsilon-\delta$.
\end{proof}

\section{Conclusion}
In this work we have introduced a framework for fine-grained reductions that can capture a deeper connection between the problems involved, namely, a correlation among their parameters. We have shown that this framework captures the essence of the fine-grained approach without restricting the results. As a byproduct of our analysis, we defined and studied the structure of improvable problems, and the implications of fine-grained reductions on such problems. Finally, we produced a fixed parameter improvement in the running time of the OV problem by utilizing its parametric correlation to the $3/2$-approx-diameter problem. 

A notable discussion in this field, is whether or not this framework can be used to define a complexity class, since FPI as a property has some unusual features. Specifically, the inherent meaning of "hardness" that arises, results in the absence of maximal elements (at least currently) in the partial ordering defined by parameterized fine-grained reductions. 
Additionally, because of the conjectured nature of our notion of improvements, the property $FPI(A,K)$ is directly related to previous work on each problem. It is possible that a parameterized algorithm may be proven sub-optimal in the case a problem's conjectured best running time is updated, resulting in disproving said property. As such, if problems having this property are considered a class, inclusion in this class could be negated after the fact, which is inconsistent with traditional complexity classes. 

Using this framework, one can follow the direction of Theorem \ref{ovtodiam} to produce parameterized improvements via the transitivity of PFGR. This analysis can be done for each reduction in the fine-grained reduction web, producing a wide variety of improved algorithms on many interesting problems.

A natural question to consider is the relation between our work and traditional parameterized approach. As seen in Theorem \ref{NPFPI}, it remains an open problem to find the exact relation between FPI and FPT, that is, to formally characterize the problems in FPT that are not FPI. Additionally, one could potentially utilize the plethora of results available through the framework on parameter tractable or harder problems. All of these results may be translated to our terminology given the appropriate assumptions.

\bibliography{refs}



\end{document}